\documentclass[12pt]{article}

\usepackage{fullpage, parskip}
\usepackage{amsmath, amsthm, amsfonts}
\usepackage{color}
\usepackage{graphicx}
\usepackage{pdfpages}

\newtheorem{theorem}{Theorem}[section]
\newtheorem{corollary}[theorem]{Corollary}
\newtheorem{lemma}[theorem]{Lemma}
\newtheorem{conjecture}[theorem]{Conjecture}

\date{}
\begin{document}

\begin{center}
	\begin{Large}{Chromatic Polynomials of Oriented Graphs}\end{Large}\\
	
	\vspace{.25in}
	
	 Danielle Cox$^{1}$ \quad \quad Christopher Duffy$^{2}$\\
	 
	 	\vspace{.15in}
	 \begin{small}
	 $^{1}$Department of Mathematics,  Mount Saint Vincent University, Halifax, CANADA\\
	 $^{2}$Department of Mathematics, University of Saskatchewan, Saskatoon, CANADA\\
	\end{small}

\end{center}

\begin{abstract}
The oriented chromatic polynomial of a oriented graph outputs the number of oriented $k$-colourings for any input $k$.
We fully classify those oriented graphs for which the oriented graph has the same chromatic polynomial as the underlying simple graph, closing an open problem posed by Sopena.
We find that such oriented graphs can be both identified and constructed in polynomial time as they are exactly the family of quasi-transitive oriented co-interval graphs.
We study the analytic properties of this polynomial and show that there exist oriented graphs which have chromatic polynomials have roots, including negative real roots,  that cannot be realized as the root of any chromatic polynomial of a simple graph.
\end{abstract}

An \emph{oriented graph} arises by assigning  directions to the edges of a simple graph.
For an oriented graph, $G$, we let $U(G)$ denote the underlying simple graph.
We say that $G$ is an \emph{orientation} of $U(G)$.
Alternately, an oriented graph is an irreflexive and  anti-symmetric digraph.

The generalization of proper colouring to graph homomorphism provides a path to  define proper colouring for oriented graphs in a way that takes into account the orientation.
For an oriented graph $G$, an \emph{oriented $k$-colouring} of $G$ is a homomorphism to a tournament (i.e., an orientation of a complete graph) on $k$ vertices.
One can see that this definition of oriented colouring is equivalent the following one, which dispenses with the need to invoke homomorphism.
For an oriented graph $G= (V_G,A_G)$, a function $c:V_G \to \{1,2,\dots k\}$ is an \emph{oriented $k$-colouring} when
	\begin{enumerate}
		\item $c(u) \neq c(v)$ for all $uv \in A_G$, and
		\item for $uv, xy \in A_G$, if $c(u) = c(y)$, then $c(v) \neq c(x)$.
	\end{enumerate}
This second condition implies directly that non-adjacent vertices at the end of a directed path of length two (a $2$-dipath) are assigned distinct colours in any oriented colouring.

Since their introduction by Courcelle in his treatment of monadic second order logic and graph structure \cite{CO94}, oriented colourings have provided a fertile area for fundamental research in mathematics and theoretical computer science.  
Many of the questions that have interested both applied and theoretical researchers in the study of graph colourings find an analogue in the study of oriented graphs.
In addition to bounds for a variety of graph families \cite{BO99, DMS18, DS14,FRR03,RASO94}, researchers have examined the computational complexity of related decision problems \cite{BHG88,KM04b}, the notion of clique for oriented graphs \cite{BDS17}, oriented arc-colourings \cite{OPS08}, oriented list-colourings \cite{T01} and even an oriented colouring game \cite{NS01}.
An excellent overview of the state-of-the-art is given in \cite{SO15}.

For an oriented graph  $G$, we define the \emph{oriented chromatic polynomial} to be the unique interpolating polynomial $f_o(G,\lambda)$ so that $f_o(G,k)$ is the number of $k$-colourings of $G$.

Figure \ref{fig:PolyExample} gives an oriented graph together with its oriented chromatic polynomial.

\begin{figure}	
	
	\begin{center}
		\includegraphics[width = 0.75\linewidth]{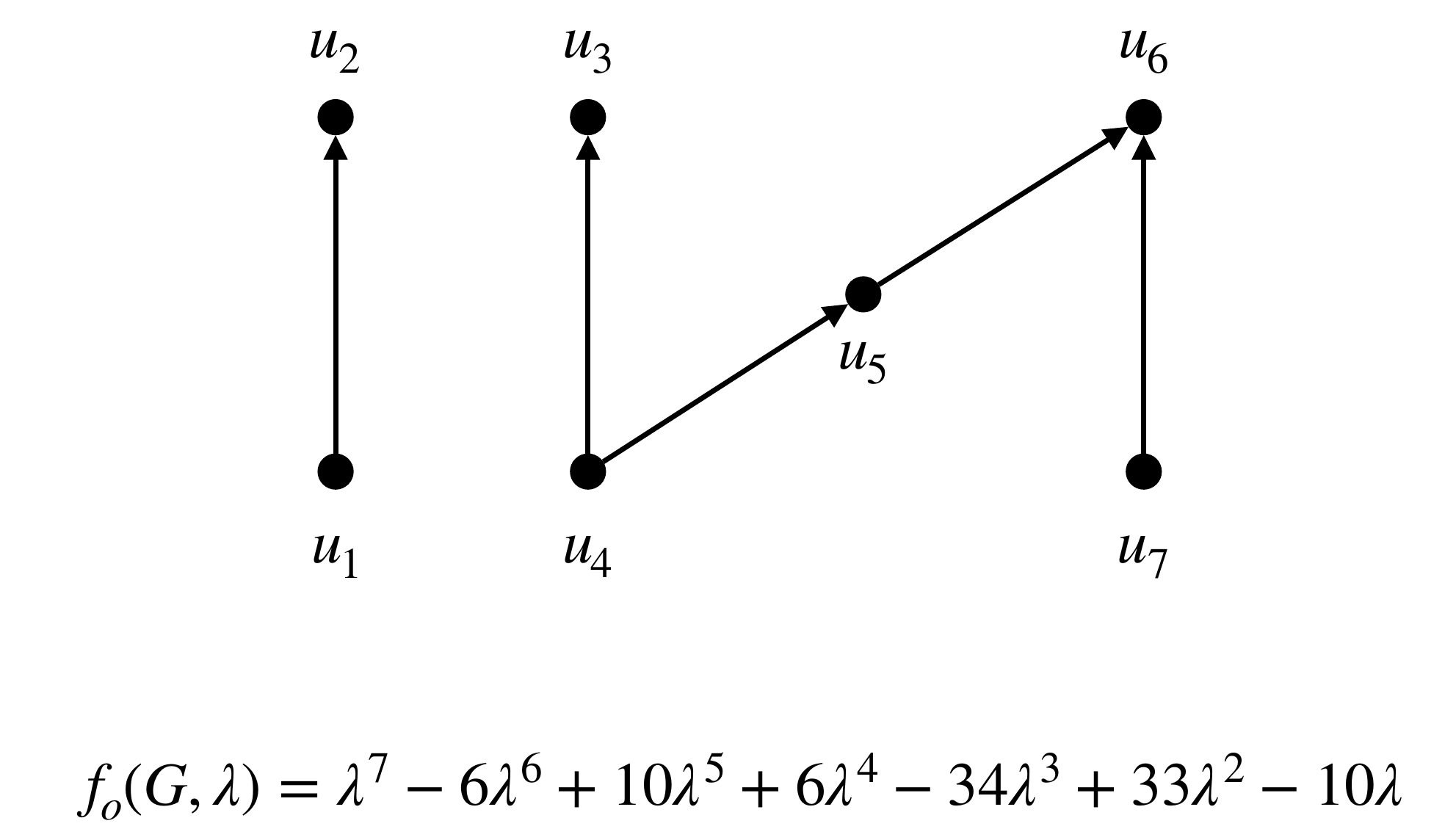}
	\end{center}
	\caption{An oriented graph and its oriented chromatic polynomial.}
	\label{fig:PolyExample}
\end{figure}

The oriented chromatic polynomial was first introduced by Sopena in \cite{S94}.
Here he established some basic properties of $f_o(G,k)$, provided a recursive construction, and exhibited some oriented graphs for which the analytic behaviour of the oriented chromatic polynomial diverged wildly from possible behaviour of the chromatic polynomial of any graph. 

The definition of the oriented chromatic polynomial follows the from definition of the \emph{chromatic polynomial}, $f(\Gamma,\lambda)$, for a graph $\Gamma$.
The polynomial $f(\Gamma,\lambda)$ is defined to be the unique interpolating polynomial so that $f(\Gamma,k)$ is the number of $k$-colourings of $\Gamma$.

We call an oriented graph \emph{chromatically invariant} when $f_o(G,\lambda) = f(U(G),\lambda)$.
One can see that the oriented graph given in Figure \ref{fig:PolyExample} is not chromatically invariant; the coefficients of $f_o(G,\lambda)$ do not alternate sign, a feature of every chromatic polynomial.
This further implies, in fact, that there is no graph $\Gamma$ such that $f_o(G,\lambda) = f(\Gamma,\lambda)$.
We say that a graph $\Phi$ and an oriented graph $H$ are \emph{chromatically equivalent} if $f_o(H,\lambda) = f(\Phi,\lambda)$.

The zeros of graph polynomials have been an active area of research for many years. One such graph polynomial whose zeroes have been studied is the chromatic polynomial. The real roots of such polynomials are dense in the interval [32/27,$\infty)$ and there are no real roots in the interval $(0,1) \cup (1,32/27]$\cite{jackson,thomassen}. 
The complex roots of such polynomials are dense in the complex plane \cite{sokalroots}.

Our work proceeds as follows. 
In the following section we review and recontextualize the recursive construction of $f_o(G,k)$ given by Sopena for the purposes of providing an explicit formula for the coefficient of $\lambda^{n-2}$ in $f_o(G,k)$. 
In Section \ref{sec:G=G} we provide a full classification of chromatically invariant oriented graphs, show they can be recognized in polynomial time and further explore the relationship between the oriented chromatic polynomial and the chromatic polynomial.
We provide a partial answer to the more general question of finding chromatically equivalent pairs of oriented graphs and graphs.
In Section \ref{sec:roots} we study the roots of oriented chromatic polynomials and show that there exist oriented graphs whose oriented chromatic polynomials have roots that cannot be realized as a root of a chromatic polynomial.
In particular we exhibit oriented graphs whose polynomials have negative real roots -- a feature of no chromatic polynomial.

All graphs considered herein are simple. 
That is, we do not allow loops or multiple edges.
Further, in graphs with more than one type of adjacency, we allow at most one type of adjacency between a pair of vertices.
We refer the reader to \cite{bondy} for  graph theoretic definitions and notation.

\section{The Oriented Chromatic Polynomial} \label{sec:OPoly}
Let $G = (V,E,A)$ be a mixed graph. 
That is, $G$ is a graph in which a subset (possibly empty) of the edges have been oriented to be arcs.
We say that $c$ is an \emph{oriented colouring} of $G$ when $c$ is an oriented colouring when $G$ is restricted to the arcs and a proper colouring when $G$ is restricted to the edges.
Notice that if every pair of vertices is either adjacent or at the ends of a directed path of length $2$, then every vertex must receive a distinct colour in every colouring.
We define the \emph{oriented chromatic polynomial} of a mixed graph analogously to that of oriented graphs.
In this section we observe that the oriented chromatic polynomial introduced by Sopena can be generalized as the oriented chromatic polynomial of mixed graphs.
This generalization allows us to find a closed form for the third coefficient of the oriented chromatic polynomial of mixed graphs.
As every oriented graph is a mixed graph with an empty arc set, this expression leads us to a closed form for the third coefficient of the chromatic polynomial of an oriented graph.
The expression for the third coefficient provides us with a tool to study chromatically invariant oriented graphs.

We begin by providing a recursive formula for the oriented chromatic polynomial of mixed graphs.
Let $G = (V_G,A_G,E_G)$ be a mixed graph.
If every pair of vertices is either adjacent or at the ends of a directed path of length $2$ (a $2$-dipath), then 

\[f_o(G,\lambda) = \prod_{i=0}^{i=n-1} (\lambda - i)\]

Otherwise, let $u$ and $v$ be a pair of vertices that are neither adjacent nor at the ends of a directed path of length two. In this case we have 

\[f_o(G,\lambda) = f_o(G+uv,\lambda) + f_o(G_{uv},\lambda),\]

where 
\begin{itemize}
	\item $G +uv$ is the mixed graph formed from $G$ by adding an edge between $u$ and $v$; and
	\item $G_{uv}$ is the mixed graph formed from $G$ by identifying $u$ and $v$ into a single vertex, deleting all parallel arcs and edges, and deleting any edge that is parallel with an arc.
\end{itemize}
Following the usual convention of having the picture of a graph stand in for its polynomial, an example of this recursion is provided in Figure \ref{fig:ReduceExample}.
\begin{figure}
	\includegraphics[width=\linewidth]{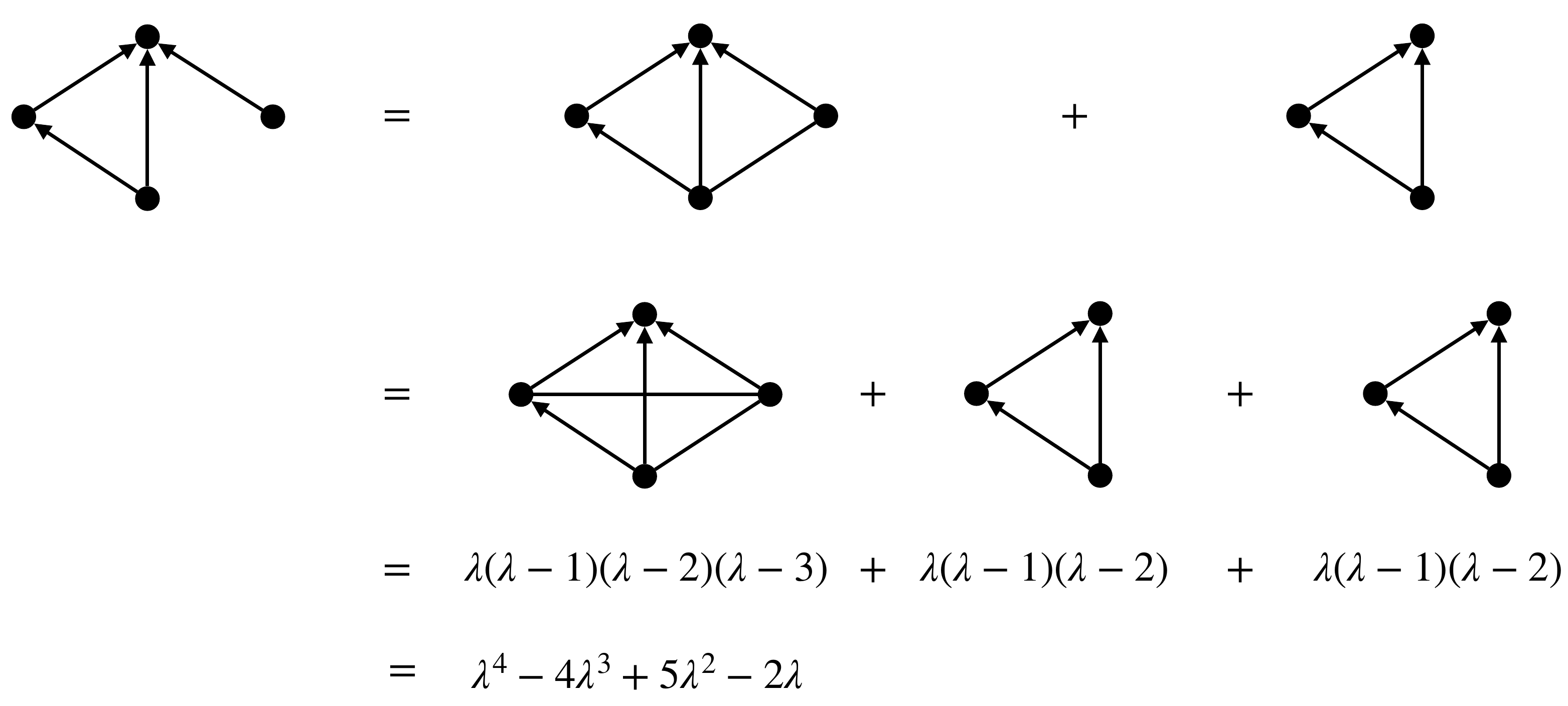}
	\caption{Computing the oriented chromatic polynomial by way of reduction.}
	\label{fig:ReduceExample}
\end{figure}

The correctness of this reduction follows from the proof of the reduction provided  in \cite{S94}.
The appendix  gives Maple code for generating the oriented chromatic polynomial of a mixed graph.

For a mixed graph $G$, let $\mathcal{D}_G$ be the set of pairs of vertices that are at the ends of an induced $2$-dipath.
In Figure~\ref{fig:PolyExample} we have $\mathcal{D}_G = \left\lbrace\{u_4,u_6\}\right\rbrace$. \\

\begin{theorem} \cite{S94} \label{thm:SopenaTheorem}
	For any mixed graph $G = (V_G,E_G,A_G)$ with $n$ vertices
	\begin{enumerate}
		\item $f(G,\lambda)$ is a polynomial of order $n$ in $\lambda$;
		\item the coefficient of $\lambda^n$ is $1$;
		\item $f(G,\lambda)$ has no constant term;
		\item the coefficient of $\lambda^{n-1}$ is $-(|A_G| + |E_G| + |
		\mathcal{D}_G|)$; and
		\item if $G$ has an isolated vertex $x$, then $f_o(G) = \lambda \cdot f_o(G-x, \lambda)$.
	\end{enumerate}
\end{theorem}

Let $G=(V_G,A_G,E_G)$ be an mixed graph with non-incident arcs $uv,xy \in A_G$. We say that the arcs $uv,xy$ are \emph{obstructing}, when
\begin{enumerate}
	\item $u$ and $y$ are not the ends of a $2$-dipath;
	\item $v$ and $x$ are not the ends of a $2$-dipath and
	\item $uy, vx \notin E_G$.
\end{enumerate}  

Let $\mathcal{O}_G$ denote the set of unordered pairs of obstructing arcs in $G$, a mixed graph.
In Figure \ref{fig:PolyExample} we have \[\mathcal{O}_G = \left\lbrace
 \{ u_1u_2,u_4u_3 \},  
 \{ u_1u_2,u_4u_5 \}, 
 \{ u_1u_2,u_5u_6 \}, 
 \{ u_1u_2,u_7u_6 \} 
  \right\rbrace.\] 

\begin{lemma}\label{lem:obstruct}
	If $\mathcal{O}_G \neq \emptyset$ or $\mathcal{D}_G \neq \emptyset$, then $f_o(G,\lambda) \neq f(U(G),\lambda)$.
\end{lemma}
\begin{proof}
	Observe that every colouring of $G$ using $k$ colours is a colouring of $U(G)$.
	However if $G$ has an induced $2$-dipath or a pair of obstructing arcs, then the converse does not hold.
	As such, there exists $k$ such that $f_o(G,k) < f(U(G),k)$.	
\end{proof}

Let $c_i(f,\lambda)$ be the coefficient of $\lambda^{n-i}$ in $f(G,\lambda)$.\\

\begin{theorem}\label{thm:Mixed3rd}
	For a mixed graph $G= (V_G,A_G,E_G)$,  we have \[ c_2(G,\lambda) = {|A_G| + |\mathcal{D}_D| + |E_G| \choose 2} - |T_G| - |\mathcal{D}_G| - |\mathcal{O}_G|,\] where $T_G$ is the set of induced subgraphs isomorphic to $K_3$ in $U(G)$.
\end{theorem}

\begin{proof}
	Let $G$ be a minimum counter-example with respect to number of vertices. Among all such counter examples, let $G$ be the one that maximizes $|A_G| + |E_G| + |\mathcal{D}_G|$.
	Note that we may further assume $\mathcal{D}_G=\emptyset$ by adding an edge between every pair of vertices in $\mathcal{D}_G$. 
	The resulting mixed graph has the same set of  oriented colourings, and thus the same oriented chromatic polynomial.
	
	The oriented chromatic polynomial of a mixed complete graph on $n$ is equal to the chromatic polynomial of a complete graph on $n$ vertices. The third coefficient of such a graph is given by ${{n \choose 3}\choose 2} - {n \choose 3}$ \cite{R68}.
	Therefore the claim holds for mixed complete graphs.
	As $G$ is a minimum counter example, $G$ is not a complete mixed graph. 
	As such there exists $u,v \in V(G)$ such that $u$ and $v$ are not adjacent, nor at the ends of a $2$-dipath. 
	Therefore $c_2(G,\lambda) = c_2(G +uv,\lambda) + c_1(G_{uv},\lambda)$. By the choice of $G$, the claim holds for both  $G +uv$ and $G_{uv}$.
	
	Let $C$ be the set of common neighbours of $u$ and $v$ in $U(G)$.
	Each of these common neighbours forms a triangle in $G+uv$.
	And so $|T_G| = |T_{G+uv}| + |C|$.
	Further observe that in $G_{uv}$, the arcs/edges from $u$ and $v$ to a common neighbour $c \in C$ becomes a single adjacency in $G_{uv}$. 
	Therefore $|A_G| + |E_G| = |A_{G_{uv}}| + |E_{G_{uv}}| - |C|$.

	A pair of obstructing arcs in $G$ is not obstructing in $G + uv$ if and only if  $u$ and $v$ are the head and tail, in some order, of the pair of obstructing arcs.
	Let $\mathcal{O}_G^{uv}$ be the set of such arcs.
	No new obstructing arcs can be created by adding an edge.
	Therefore  $|\mathcal{O}_G| - |\mathcal{O}_G^{uv}|$  = $|\mathcal{O}_{G + uv}|$.
	Also note that a pair of obstructing arcs in $|\mathcal{O}_G^{uv}|$ form a $2$-dipath in $G_{uv}$ with centre vertex $uv$.
	All other induced $2$-dipaths in $G$ are retained in $G_{uv}$, as $u$ and $v$ are not the ends of a $2$-dipath.
	Therefore $|D_G| + |\mathcal{O}_G^{uv}| = |D_{G_{uv}}|$.
	
	\begin{equation}
	c_2(G,\lambda) = c_2(G +uv,\lambda) + c_1(G_{uv},\lambda)\\
	\end{equation}
	\[ ={|A_G| + |D_G| + |E_G| + 1 \choose 2} - (|T_{G+uv}| + |D_{G + uv}| + |\mathcal{O}_{G + uv}|) \]
	\[- ( |A_{G_{uv}}| + |E_{G_{uv}}| + |D_{G_{uv}}|) \]
	\[= {|A_G| + |D_G| + |E_G| + 1 \choose 2} - (|T_G| - |C| + |D_G| + |\mathcal{O}_G| - |\mathcal{O}_G^{uv}| )\]
	\[- ( |A_{G}| + |E_{G}| + |C| + |D_{G}| + |\mathcal{O}_G^{uv}|)\]
	
	\[ =  {|A_G| + |D_G| + |E_G| + 1 \choose 2} - (|A_G| + |D_G| + |E_G|) - |T_G| - |D_G| -  |\mathcal{O}_G|   \]
	
\[ =  {|A_G| + |D_G| + |E_G| \choose 2} - |T_G| - |D_G| - |\mathcal{O}_G.|\]

Thus the claim holds for $G$, contradicting the choice of $G$ as a minimum counter example.
\end{proof}

For the case $E = \emptyset$, we arrive at the desired result for oriented graphs.\\

\begin{corollary}\label{cor:Orient3rd}
For an oriented graph $G$, we have \[c_2(G,\lambda) = {|A_G| + |\mathcal{D}_G| \choose 2} - |T_G| - |\mathcal{D}_G| - |\mathcal{O}_G|. \]
\end{corollary}

We further note that in the case $A = \emptyset$, we arrive at the usual result for the third coefficient of the chromatic polynomial: $c_2(G, \lambda) = {|E_G| \choose 2} - |T_G|$.


\section{Oriented Chromatic Equivalence} \label{sec:G=G}
A folklore construction gives an orientation $G$ of $K_{n,n}$ so that the resulting oriented graph has chromatic number $2n$ (see \cite{DMS18}).
This common example is used to convince the reader that the oriented chromatic number and the chromatic number of the underlying simple graph can be arbitrarily far apart.
We note, however that the set of colourings of $G$ using $\lambda \geq 2n$ colours is exactly that of colourings of $K_{2n}$ using $\lambda$ colours.
And so though the chromatic number of $U(G)$ differs greatly to $G$, there is still a relationship between colourings of $G$ and colourings of some simple graph.
In this section we find a set of sufficient conditions so that the $\lambda$-colourings of an oriented graph $G$ are exactly those of some simple graph $\Gamma$.
We conclude this section by using these sufficient conditions to compute the chromatic polynomial of orientations of stars.

 We are interested in the following decision problems:

\underline{\emph{CHROM-INVAR}}\\
\quad \emph{Instance:} A graph $\Gamma$.\\
\quad \emph{Question:} Is there an orientation $O(\Gamma)$ such that $f_o(O(\Gamma),\lambda) = f(\Gamma,\lambda)$?\\

\underline{\emph{OCHROM-INVAR}}\\
\quad \emph{Instance:} An oriented graph $G$.\\
\quad \emph{Question:} Does $f_o(G,\lambda) = f(U(G),\lambda)$?\\

\underline{\emph{OCHROM-EQUIV}}\\
\quad \emph{Instance:} An oriented graph $G$.\\
\quad \emph{Question:} Is there a  graph  $\Gamma$ such that $f_o(G,\lambda) = f(\Gamma,\lambda)$?\\

Figure~\ref{fig:EquivExample} gives an example of an oriented graph $G$ and a graph $\Gamma$ so that $G$ and $\Gamma$ are chromatically equivalent.

\begin{figure}	
	\begin{center}
		\includegraphics[width = \linewidth]{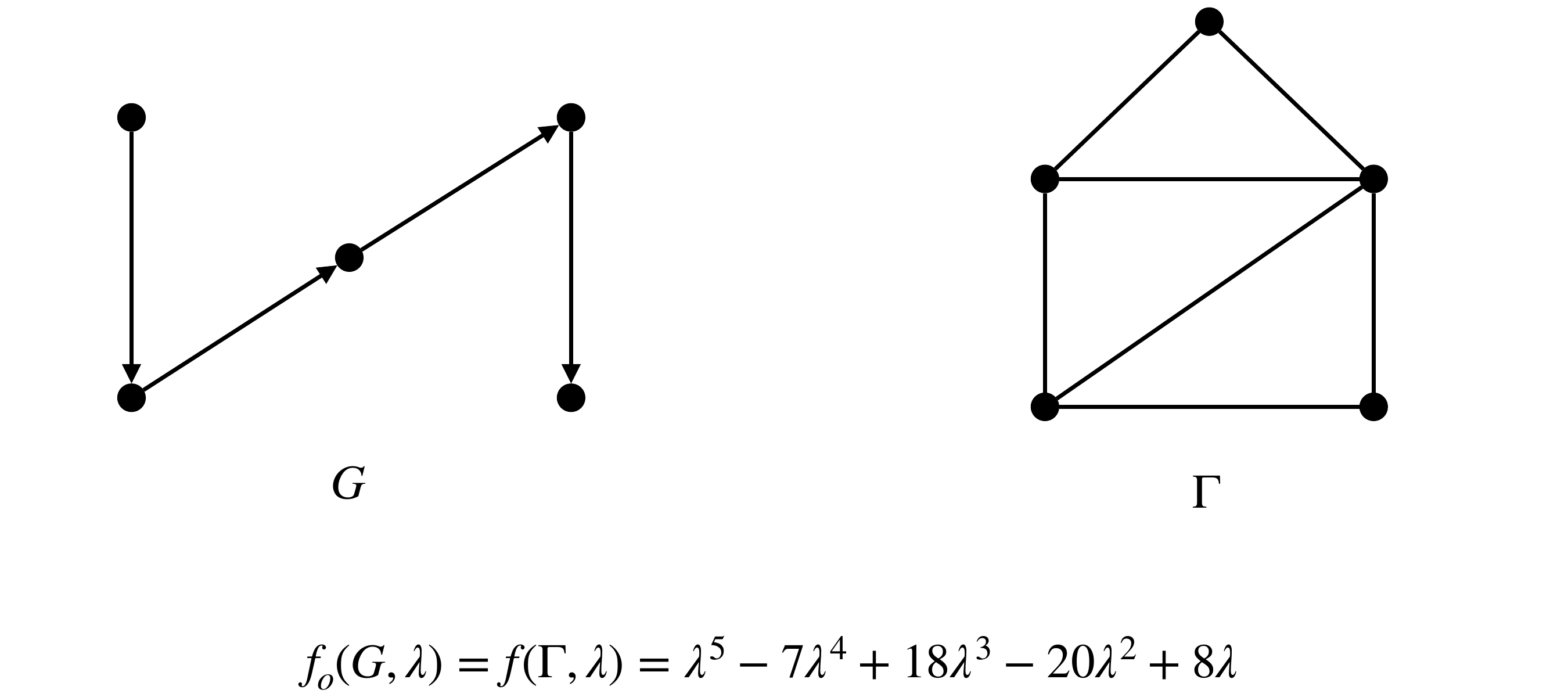}
	\end{center}
	\caption{An oriented graph and a graph with the same chromatic polynomial.}
	\label{fig:EquivExample}
\end{figure}

Let $G$ be an oriented graph.  Let $G^\star$ be the mixed graph resulting from $G$ by adding an edge between $u$ and $v$ whenever $u$ and $v$ are at the ends of an induced $2$-dipath.
We observe the following.\\

\begin{lemma}\label{lem:starEqual}
	$f_o(G,\lambda) = f_o(G^\star,\lambda)$.
\end{lemma} 

\begin{proof}
	$G$ is a subgraph of $G^\star$.
	Therefore $f_o(G,\lambda) \leq  f_o(G^\star,\lambda)$.
	Every oriented colouring of $G^\star$ using $k$ colours is also an oriented colouring of $G$, therefore $f_o(G,\lambda) \geq f_o(G^\star,\lambda)$.\\
\end{proof}

\begin{theorem}\label{thm:squarethm}
	For $G$, an oriented graph,  $f_o(G,\lambda) = f(U(G^\star),\lambda)$ if and only if $\mathcal{O}_G = \emptyset$.
\end{theorem}

\begin{proof}
	By Lemma \ref{lem:starEqual} it suffices to show $f_o(G^\star,\lambda) = f(U(G^\star),\lambda)$.
	Notice that $G$ has no obstructing arcs if and only if $G^\star$ has no obstructing arcs.
	By Lemma \ref{lem:obstruct} it suffices to show that if $G^\star$ has no obstructing arcs, then every colouring of $U(G^\star)$ is an oriented colouring of $G$.
	Let $c$ be a colouring of $U(G^\star)$.
	Since $c$ is a proper colouring, if $c$ is not an oriented colouring of $G^\star$, then the second condition of oriented colouring as been violated.
	However, this not possible as $G^\star$ has neither an induced $2$-dipath nor a pair of obstructing arcs.
\end{proof}

Recall the result of Corollary \ref{cor:Orient3rd}.
If $G$ has no pair of obstructing arcs, then in $G^\star$ we have $\mathcal{O}_G =\mathcal{D}_G = \emptyset$. 
And so $c_2(G^\star,\lambda) = {|A_G| + |E_G| \choose 2} - |T_G|$.
Notice that this is exactly the third coefficient of the chromatic polynomial of $U(G^\star)$.\\

\begin{corollary}\label{cor:G=G}
	An oriented graph $G$ is chromatically invariant if and only if $G$ has no induced $2$-dipath and $U(G)$ is $2K_2$-free.
\end{corollary}

\begin{proof}
	Assume $G$ has no induced $2$-dipath and that $U(G)$ is $2K_2$-free.
	It follows directly that $G$ has no obstructing arcs and that $G = G^\star$.
	The conclusion follows by Theorem \ref{thm:squarethm}.
	
	Let $G$ be a chromatically invariant oriented graph.
	By definition, every proper $k$-colouring of $U(G)$ is an oriented colouring of $G$.
	Therefore $G$ has no induced $2$-dipath, nor does $G$ contain a pair of obstructing arcs.
	Since $G$ contains no pair of obstructing arcs, if $U(G)$ contains an induced copy of $2K_2$, say $uv, xy$, then without loss of generality, there must be an induced $2$-dipath between $u$ and $y$. 
	This contradicts that $G$ contains no induced $2$-dipath. 
	Therefore $G$ has no induced $2$-dipath and $U(G)$ is $2K_2$-free.
\end{proof}

Introduced by Ghouila-Houri, oriented graphs that contain no induced $2$-dipath are called  \emph{quasi-transitive}.\\

\begin{theorem}\cite{GH62}\label{thm:qtGraphs}
	A graph $\Gamma$ admits a quasi-transitive orientation if and only if $\Gamma$ is a comparability graph.
\end{theorem}

Notice that the family of $2K_2$-free comparability graphs is exactly the family of co-interval graphs. And so combining Corollary \ref{cor:G=G} and Theorem \ref{thm:qtGraphs} yields the following classification.\\

\begin{theorem}\label{thm:everyOrient}
An oriented graph is chromatically invariant if and only if it is a quasi-transitive orientation of a co-interval graph.
\end{theorem}

This theorem fully classifies chromatically invariant oriented graphs, as well as those graphs that admit a chromatically invariant orientation.
This closes an open problem posed by Sopena in \cite{S94} and provides a geometric interpretation of chromatically invariant oriented graphs as co-interval graphs.

Theorem \ref{thm:everyOrient} implies that the decision problems \emph{CHROM-INVAR} and \emph{OCHROM-INVAR} can be restated in terms of  co-interval graph recognition.
As co-interval graphs can be identified in linear time \cite{S03}, we arrive at the following classification  of  \emph{CHROM-INVAR} and \emph{OCHROM-INVAR}.

\begin{corollary}\label{thm:complexity}
	The decision problems \emph{CHROM-INVAR} and \emph{OCHROM-INVAR}  are \emph{Polynomial}.
\end{corollary}

We note that such an orientation of a co-interval graph need not be unique (up to converse).
There are many methods in the literature (for example see \cite{B08}) that give a quasi-transitive orientation of a comparability graph.
A common element of these methods is the construction of an auxiliary graph, $Aux(G)$, so that a $2$-colouring of $G$ corresponds to a quasi-transitive orientation. 
Such constructions imply that a comparability graph has a unique quasi-transitive orientation (up to converse) if and only if $Aux(G)$ is connected.
As such constructions can be carried out in polynomial time, we find that given a co-interval graph $\Gamma$, one may find in polynomial time an orientation of $\Gamma$, $O(\Gamma)$, so that $f(\Gamma,\lambda) = f_o(O(\Gamma), \lambda)$.

%
%

We now consider an application of Theorem \ref{thm:squarethm} and find the oriented chromatic polynomial of the family of orientations of stars. Let $S_{i,o}$ be the orientation of a star on $i+o+1$ vertices, with centre vertex $x$, so that $x$ has $i$ in-neighbours and $o$ out-neighbours. \\

\begin{corollary}
	$f_o(S_{i,o}) = \lambda\cdot f(K_{i,o},\lambda-1)$
\end{corollary}

\begin{proof}
	Observe that $S_{i,o}^\star$ has no obstructing arcs. 
	Further observe that  $U(S_{i,o}^\star)$ consists of a copy of $K_{i,o}$ together with a universal vertex.
	By Theorem \ref{thm:squarethm}, we have $f_o(S_{i,o}) = \lambda \cdot f(K_{i,o},\lambda-1)$.\\
\end{proof}

Conversely, using the results from Section \ref{sec:OPoly}, one can find families of oriented graphs for which there is no  chromatically equivalent graph.
Let $G$ be an orientation of $tK_2$ for some $t > 1$.
Recalling the notation of the previous section we have
\begin{itemize}
	\item $|A_G| = t$;
	\item $|E_G| = |\mathcal{D}_G| = 0$;
	\item $|\mathcal{O}_G| = {t \choose 2}$; and
	\item $|T_G| = 0$.
\end{itemize}

If there exists $\Gamma$ such that $f_o(G,\lambda) = f(\Gamma,\lambda)$, then by Theorem \ref{thm:SopenaTheorem} and Corollary \ref{cor:Orient3rd}, it must be that $\Gamma$ has $2t$ vertices, $t$ edges and ${t \choose 2}$ copies of $K_3$.
A simple counting argument implies that no such $\Gamma$ can exist.\\

\begin{conjecture}\label{conj:G=H}
	Let $G$ be an oriented graph. There is a graph $\Gamma$ such that $G$ and $\Gamma$  are chromatically equivalent if and only if $G$ has no obstructing pairs of arcs.
\end{conjecture}

When restricted to oriented graphs that contain no pair of obstructing arcs, \emph{OCHROM-EQUIV} is Polynomial -- every instance is a YES instance.
However, for arbitrary inputs is not clear if \emph{OCHROM-EQUIV} is contained in NP, as constructing the chromatic polynomial of a graph is NP-hard.
We conjecture, however, that those oriented graphs with no obstructing arcs are the only oriented graphs whose $\lambda$-colourings have a one-to-one correspondence with the $\lambda$-colourings of some graph $\Gamma$.
If Conjecture \ref{conj:G=H} is true, then \emph{OCHROM-EQUIV} is Polynomial for arbitrary inputs.

Conjecture \ref{conj:G=H} is not true when we allow $G$ to be a mixed graph, even when we require $A \neq \emptyset$.
Let $H$ be the mixed graph formed from a pair of disjoint arcs by adding an edge between the heads and an edge between the tails.
Using the reduction given in Section \ref{sec:OPoly}, we find $f_o(H,\lambda) = \lambda^4 - 4\lambda^3 + 5\lambda^2 -2\lambda$.
Notably, this is exactly the chromatic polynomial of the example, $G$, given in Figure \ref{fig:ReduceExample}.
The oriented graph $G$ has $\mathcal{D}_G = \emptyset$.
Therefore $G = G^\star$.
Further $\mathcal{O}_G = \emptyset$.
And so by Theorem \ref{thm:squarethm} we have  $f(U(G),\lambda) = f_o(G,\lambda) = f_o(H,\lambda)$.
In particular, $H$ has a pair of obstructing arcs, but yet there is graph $\Gamma$ such that $f_o(H,\lambda) = f(\Gamma,\lambda)$.
From this example, one may generate an example on $n$ vertices for any $n \geq 4$ by repeatedly adding universal vertices to $H$ and $U(G)$.

\section{Roots of Oriented Chromatic Polynomials}\label{sec:roots}
The location of the roots of polynomials has been well studied for a variety of graph polynomials, such the independence, domination, reliability and chromatic polynomials. 
In this section we  provide  results regarding the roots of the oriented chromatic polynomial. Chromatic polynomials have roots that are dense in the complex plane \cite{sokalroots}. 
Their coefficients alternate in sign and hence have no negative real roots \cite{read}. 
We show the following:

\begin{theorem}
For every integer $k > 0$, there exists an oriented graph $G$ so that $f_o(G,\lambda)$ has a root $k^\prime$ so that $k^\prime < -k$.

\end{theorem}
\begin{proof}

Let $D_n$ be the oriented graph on $n$ vertices that consists of a directed path $v_1,v_2,v_3,v_4$ with $n-4$ leaves, $x$ directed from $x$ to $v_4$.

The oriented chromatic polynomial can be computed in the following way. Vertex $v_1$ has $\lambda$ colour choices, $v_2$ has $\lambda -1$ colour choices and $v_3$ has $\lambda-2$ colour choices. 
Now for $v_4$ we have two options. If $v_4$ is the same colour as $v_1$ then the $n-4$ leaves have $\lambda-2$ choices of colour, since it can not be the same colour at $v_4$ or $v_2$. 
If $v_4$ is a different colour than the others in the path it has $\lambda-3$ colour choices and the leaves have $\lambda -1$ colour choices. 
This means 
\[ f_o(D_n,\lambda)=\lambda(\lambda-1)(\lambda-2)((\lambda-2)^{n-4}+(\lambda-3)(\lambda-1)^{n-4})\]

The polynomial $f_o(D_n,\lambda)$ has real roots at $\lambda=0,1,2$. 
We show  we can obtain arbitrarily large negative real roots by showing that a real root exists between $\lambda=-n$ and $\lambda=-\ln(n)$ for $n$ even.

Observe,  
\[ f_o(D_n,-n)=(-1)^3n(n+1)(n+2)((-1)^{n-4}(n+2)^{n-4}+(-1)^{n-3}(n+3)(n+1)^{n-4})\]

It can be shown that $(n+2)^{n-4}<(n+3)(n+1)^{n-4}$ if 
\[ \Big( 1+\frac{1}{n+1} \Big)^{n-4}<n+3.\] 
The quantity $\Big( 1+\frac{1}{n+1} \Big)^{n-4}$ is bounded above by $e$, and $e<n+3$ for all values of $n$, therefore $f_o(D_n,-n)>0$ for all even values of $n$.

Now consider $f_o(D_n,-\ln(n))$. 
\begin{eqnarray*}
f_o(D_n,-\ln(n))&=&(-1)^3(\ln(n))(\ln(n)+1)(\ln(n)+2)*\\
&&\Big((-1)^{n-4}(\ln(n)+2)^{n-4}+(-1)^{n-3}(\ln(n)+3)(\ln(n)+1)^{n-4}\Big)
\end{eqnarray*}

Clearly $(\ln(n)+2)^{n-4}>(\ln(n)+3)(\ln(n)+1)^{n-4}$ when 
\[ \Big( 1+\frac{1}{\ln(n)+1}\Big)^{n-4} > \ln(n)+3.\]

Let $g(n)=\Big( 1+\frac{1}{\ln(n)+1}\Big)^{n-4} -\ln(n)-3$. The derivative of this function is\\
$g^{\prime}(n)=\left( 1+ \left( \ln  \left( n \right) +1
 \right) ^{-1} \right) ^{n-4} \left( \ln  \left( 1+ \left( \ln 
 \left( n \right) +1 \right) ^{-1} \right) -{\frac {n-4}{ \left( \ln 
 \left( n \right) +1 \right) ^{2}n \left( 1+ \left( \ln  \left( n
 \right) +1 \right) ^{-1} \right) }} \right) -\frac{1}{n}$.

It is the case that $\lim_{n\to \infty}g'(n)=\infty$, thus there exists $N$ so that for all $n>N$, $g(n)$ is an increasing function, as the derivative  of $g$ is positive and hence $g(n)>0$ and $f_o(D_n,-\ln(n))<0$ for large values of $n$.

It then follows by the intermediate value theorem that  $f_o(D_n,\lambda)$ can have an arbitrarily large negative root.
\end{proof}

A chromatic polynomial cannot have root in the interval $(-\infty,0) \cup (0,1) \cup (1,\frac{32}{27})$ \cite{read}. 
We have shown that oriented chromatic polynomials can have negative real roots. 
In addition, there exist oriented chromatic polynomials in the interval $(0,1)$, as $f_o(D_5,\lambda)$ has a root at $\lambda=\frac{3}{2}-\frac{\sqrt{5}}{2}$.
Open problems regarding the roots of oriented chromatic polynomials include: does there exist an oriented graph whose real roots lie in $(1,\frac{32}{27})$? What is the closure of the complex roots for the oriented chromatic polynomial?


\section{Conclusion}
The study of oriented graphs often goes hand-in-hand with that of signed graphs.
Though the methods contained herein will extend to the study of chromatic polynomials of signed graphs, there will be a marked difference in the classification of chromatically invariant signed graphs.
For example, letting all the edges of  $2K_2$ be positive leads to a chromatically invariant signed graph. 
However, all possible orientations of $2K_2$ leads to an oriented graph that is not chromatically invariant.
Of course, with this approach, every graph can have edge signs trivially assigned so that the resulting signed graph has the same chromatic polynomial as the underlying graph.
And so one may require that there is at least one edge of each sign.
With this added restriction it is unclear if chromatically invariant signed graphs can be identified in polynomial time, as we expect the characterization to require that signs be given so that there is no $2K_2$ where the edges have different signs.
Similarly, the generalization of signed graphs and oriented graphs to $(m,n)$-mixed graphs should yield a definition of a chromatic polynomial that obeys the reduction outlined in Section \ref{sec:OPoly}. 
Consequently we expect the results of Theorems \ref{thm:SopenaTheorem} and \ref{thm:Mixed3rd} to generalize in the same manner.
One may also ask, then, for which graphs $\Gamma$ is there an assignment of arcs, edges, and corresponding colours, so that the resulting $(m,n)$-mixed graph is chromatically invariant.
We have shown  that for oriented graphs that $\Gamma$ must be a co-interval graph.

\section*{Acknowledgments}
The authors thank Gary MacGillivray for discussions regarding the recursive construction of the oriented chromatic polynomial.

\bibliographystyle{abbrv}
\bibliography{references.bib}
\section*{Appendix}


\section*{Supplementary material}

```
#Takes as an input an oriented graph G, a set of edges E and s = |V(G)|+1
#Outputs the oriented chromatic polynomial of G with edge set E

ocpoly := proc (G, E, s)
    local i, A, V, locals, p2, H, EH, newA, E1, newGarcs, newEdges, a, a1,
a2, fl, EG, newG, p3;
    option remember;
    p2 := 0;
    locals := s+1; #Vertex label for vertex created when identifying a pair
of vertices.
    H := UnderlyingGraph(G);
    V := Vertices(G);
    EH := Edges(H);
    newA := choose(convert(V, set), 2) minus EH minus E; #Pairs of non-
adjacent vertices the mixed graph.
    E1 := E;
    newGarcs := {};
    newEdges := {};

    #If there are no pairs of non-adjacent vertices, then output the falling
factorial.
    if nops(newA) = 0 then
        p2 := product(x-i, i = 0 .. nops(V)-1);
        return p2;

    #Otherwise, we find u and v, a pair of non-adjacent vertices. We
construct G+uv and G_uv and find their chromatic polynomial.
    else
        a := op(1, newA);
        a1 := op(1, a);
        a2 := op(2, a);
        fl := 0;

        #Check if u and v are at the ends of a 2-dipath
        if convert(Departures(G, a1), set) intersect convert(Arrivals(G,
a2), set) != {} then
            fl := 1;
        end if;

        if convert(Arrivals(G, a1), set) intersect convert(Departures(G,
a2), set) != {} then
            fl := 1;
        end if;

        #If u and v are not the ends of a 2-dipath, then construct G_uv and
G+uv by building new arc/edge sets.
        if fl = 0 then
            EG := Edges(G);

            for i to nops(EG) do
                if op(1, op(i, EG)) = a1 then
                    newGarcs := newGarcs union {[s, op(2, op(i, EG))]};
                elif op(1, op(i, EG)) = a2 then
                    newGarcs := newGarcs union {[s, op(2, op(i, EG))]};
                elif op(2, op(i, EG)) = a1 then
            newGarcs := newGarcs union {[op(1, op(i, EG)), s]};
                elif op(2, op(i, EG)) = a2 then
                    newGarcs := newGarcs union {[op(1, op(i, EG)), s]};
                else
                    newGarcs := newGarcs union {convert(op(i, EG), list)};
```

```
                    end if
                end do;

            for i to nops(E1) do
                if op(1, op(i, E1)) = a1 then
                    newEdges := newEdges union {{op(2, op(i, E1)), s}};
                elif op(1, op(i, E1)) = a2 then
                    newEdges := newEdges union {{op(2, op(i, E1)), s}};
                elif op(2, op(i, E1)) = a1 then
                    newEdges := newEdges union {{op(1, op(i, E1)), s}};
                elif op(2, op(i, E1)) = a2 then
                    newEdges := newEdges union {{op(1, op(i, E1)), s}};
                else
                    newEdges := newEdges union {op(i, E1)};
                end if
            end do;

            newG := Digraph(newGarcs);
            p2 := ocpoly(newG, newEdges, locals)

        #If u and v are the ends of a 2-dipath, then return 0 for G_uv
        else
            p2 = 0;
        end if;

        E1 := E1 union {{a1, a2}};
        p3 := ocpoly(G, E1, locals);
        return (p2+p3);
    end if;
end proc;
```


\end{document}